\documentclass{amsart}

\usepackage{hyperref, amsmath, amsthm}
\usepackage[authoryear]{natbib}
\usepackage{amsaddr}
\usepackage{graphicx}
\usepackage{subcaption}
\usepackage{tikz}
\usetikzlibrary{arrows,positioning,shapes.geometric}

\newtheorem{theorem}{Theorem}
\newtheorem{proposition}{Proposition}
\newtheorem{lemma}{Lemma}

\newenvironment{sketchproof}[1][\proofname]{\proof[#1]\mbox{}}{\endproof}

\newtheoremstyle{propertystyle}
{3pt} 
{3pt} 
{\it} 
{} 
{\bfseries} 
{.} 
{.5em} 
{} 
\theoremstyle{propertystyle}
\newtheorem{property}{Property}

\newtheorem*{assumption*}{\assumptionnumber}
\providecommand{\assumptionnumber}{}
\makeatletter
\newenvironment{assumption}[2]
 {%
  \renewcommand{\assumptionnumber}{Assumption #1 (#2)}%
  \begin{assumption*}%
  \protected@edef\@currentlabel{#1 (#2)}%
 }
 {%
  \end{assumption*}
 }
\makeatother

\theoremstyle{remark}

\begin{document}

\title[Entropy Balancing]{Entropy Balancing is Doubly Robust}
\author{Qingyuan Zhao}
\address[Qingyuan Zhao]{Department of Statistics, Wharton School,
  University of Pennsylvania}
\email{qyzhao@wharton.upenn.edu}
\author{Daniel Percival}
\address[Daniel Percival]{Google Inc.}
\email{dancsi@google.com}
\date{\today}

\begin{abstract}
Covariate balance is a conventional key diagnostic for methods used
estimating causal effects from observational studies. Recently, there
is an emerging interest in directly incorporating covariate balance in the
estimation. We study a recently proposed entropy maximization method
called Entropy Balancing (EB), which exactly matches the covariate
moments for the different experimental groups in its optimization problem.
We show EB is doubly robust with
respect to linear outcome regression and logistic propensity score
regression, and it reaches the asymptotic semiparametric
variance bound when both regressions are correctly specified. This is
surprising to us because there is no
attempt to model the outcome or the treatment assignment in the
original proposal of EB. Our theoretical results and simulations
suggest that EB is a very appealing alternative to the conventional
weighting estimators that estimate the propensity score by maximum likelihood.

\end{abstract}

\thanks{This work is completed when Qingyuan Zhao is a Ph.D.\ student
  at the Department of Statistics, Stanford University. We would like to thank Jens Hainmueller, Trevor Hastie, Hera He, Bill
  Heavlin, Diane Lambert, Daryl Pregibon, Jean Steiner and one
  anonymous reviewer for their helpful comments.}


\maketitle

\section{Introduction}
\label{sec:introduction}

Consider a typical setting of observational study that two
conditions (``treatment'' and ``control'') are not randomly
assigned to the units. 
Deriving a causal conclusion from such observational data is
essentially difficult because the treatment exposure may be related to
some covariates that are also related to the outcome. In this case, those
covariates may be imbalanced between the treatment groups and the
naive mean causal effect estimator can be severely biased.

To adjust for the covariate imbalance, the seminal work of
\citet{rosenbaum1983} points out the essential role of \emph{propensity
score}, the probability of exposure to treatment conditional on
observed covariates. This quantity, rarely known in an observation
study, may be estimated from the data. Based on the estimated
propensity score, many statistical methods are proposed to estimate
the mean causal effect. The most popular approaches are
\emph{matching}
\citep[e.g.][]{rosenbaum1985constructing,abadie2006large},
\emph{stratification} \citep[e.g.][]{Rosenbaum1984}, and \emph{weighting}
\citep[e.g.][]{Robins1994,Hirano2001}. Theoretically, propensity
score weighting is the most attractive among these
methods. \citet{Hirano2003} show that nonparametric propensity score
weighting can achieve the semiparametric efficiency bound for the
estimation of mean causal effect derived by \citet{Hahn1998}. Another
desirable property is double robustness. The pioneering work of
\citet{Robins1994} augments propensity score weighting by an outcome
regression model. The resulting estimator has the so-called
\emph{double robustness} property:
\begin{property}
\label{property:dr}
If either the propensity score model or the outcome regression model is
correctly specified, the mean causal effect estimator is statistically consistent.
\end{property}

In practice, the success of any propensity score method hinges on the
quality of the estimated propensity score.
The weighting methods are usually more sensitive to model misspecification
than matching and stratification, and furthermore, this bias can even
be amplified by a doubly robust estimator, brought to
attention by \citet{Kang2007}.
In order to avoid model misspecification, applied researchers usually increase the
complexity of the propensity score model until a sufficiently
balanced solution is found. This cyclical process of modeling
propensity score and checking covariate balance is criticized as the
``propensity score tautology'' by \citet{Imai2008} and, moreover, has
no guarantee of finding a satisfactory solution eventually.

Recently, there is an emerging interest, particularly among applied
researchers, in directly incorporating
covariate balance in the estimation procedure, so there is no
need to check covariate balance repeatedly
\citep[e.g.][]{graham2012inverse,Diamond2013,Imai2014,zubizarreta2015stable}.
In this paper, we study a method of this kind called \emph{Entropy
  Balancing} (hereafter EB) proposed in \citet{Hainmueller2011}.
In a nutshell, EB solves an (convex) entropy
maximization problem under the constraint of exact balance of
covariate moments. Due to its easy interpretation and fast
computation, EB has already gained some popularity in applied fields
\citep{Marcus2013,Ferwerda2014}. However, little do we known about the
theoretical properties of EB. The original proposal in
\citet{Hainmueller2011} did not give a condition such that EB is
guaranteed to give a consistent estimate of the mean causal effect.

In this paper, we shall show EB is indeed a very appealing propensity score weighting
method. We find EB simultaneously fits a logistic regression model for
the propensity score and a linear regression model for the
outcome. The linear predictors of these regression models are the
covariate moments being balanced. We shall prove EB is doubly robust
(Property \ref{property:dr}), in the sense that if at least one of the
two models are correctly specified, EB is consistent for the Population
Average Treatment effect for the Treated (PATT), a common quantity of
interest in causal inference and survey sampling.
Moreover, EB is sample bounded \citep{Tan2010}, meaning the PATT
estimator is always within the range of the observed outcomes, and it
is  semiparametrically efficient if both models are correctly
specified. Lastly, The two linear models have an exact correspondence
to the primal and dual optimization problem used to solve EB,
revealing an interesting connection between doubly robust estimation
and convex optimization.

Our discoveries can be summarized in the diagram in Figure
\ref{fig:ps-or-cb}. Conventionally, the recipe given by Robins and
his coauthors is to fit separate models for propensity score and
outcome regression and then combine them by a doubly robust estimator
\citep[see e.g.][]{Robins1994,Lunceford2004,Bang2005,Kang2007}. In contrast, Entropy Balancing
achieves this goal through enforcing covariate balance. The
primal optimization problem of EB amounts to
an empirical calibration estimator \citep{deville1992calibration,sarndal2005estimation},
which is widely popular in survey sampling but perhaps not sufficiently recognized in
causal inference \citep{chan2015}. The balancing constraints in this optimization
problem result in unbiasedness of the PATT estimator under linear
outcome regression model. The dual optimization problem of EB is
fitting a logistic propensity score model with a loss function
different from the negative binomial likelihood. The
Fisher-consistency of this loss function (also called \emph{proper}
scoring rule in statistical decision theory, see e.g.\
\citet{gneiting2007strictly}) ensures the other half of double
robustness---consistency under correctly specified propensity score
model. Since EB essentially just uses a different loss function, other
types of propensity score models, for example the generalized additive
models \citep{hastie1990generalized}, can also easily be fitted. A
forthcoming article by \citet{zhao2016covariate} offers more discussion
and extension to other weighted average treatment effects.

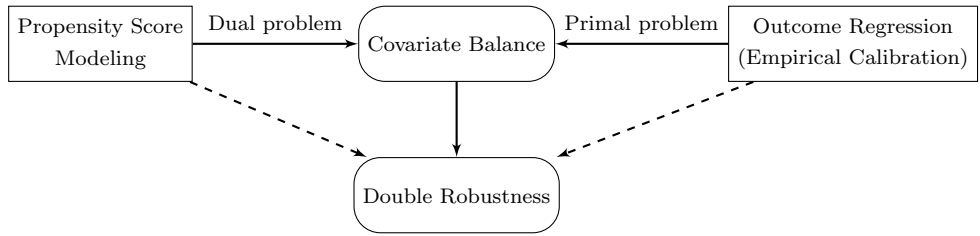
\begin{figure}[t]
  \centering
\begin{tikzpicture}[>=latex']
        \tikzset{block/.style= {draw, rectangle, align=center,minimum width=2cm,minimum height=1cm},
        rblock/.style={draw, shape=rectangle,rounded corners=1.0em,align=center,minimum width=2cm,minimum height=1cm},
        input/.style={ 
        draw,
        trapezium,
        trapezium left angle=60,
        trapezium right angle=120,
        minimum width=2cm,
        align=center,
        minimum height=1cm
    },
        }
        \node [block]  (ps) {\footnotesize Propensity Score \\
          \footnotesize Modeling};
        \node [rblock, right = 2.2cm of ps] (cb) {\footnotesize Covariate Balance};
        \node [block, right = 2.3cm of cb] (or) {\footnotesize Outcome
          Regression \\ \footnotesize (Empirical Calibration)};
        \node [rblock, below = 1cm of cb] (br) {\footnotesize Double Robustness};

        \path[draw,->,thick] (ps) edge node[auto, midway] {\footnotesize Dual problem} (cb)
                    (or) edge node[auto, midway, above] {\footnotesize Primal problem} (cb)
                    (cb) edge (br)
                    ;
        \path[draw,dashed,->,thick] (ps) edge (br)
                                    (or) edge (br)
                    ;
    \end{tikzpicture}
    \caption{The role of covariate balance in doubly robust
      estimation. Dashed arrows: conventional procedure to achieve
      double robustness. Solid arrows: double robustness of Entropy
      Balancing via covariate balance.}
    \label{fig:ps-or-cb}
\end{figure}


\section{Setting}
\label{sec:causal-model-doubly}


First, we fix some notations for the causal inference problem considered in this paper.
We follow the potential outcome language of
\citet{neyman1923applications} and \citet{Rubin1974}.
In this causal model, each unit $i$
is associated with a pair of potential outcomes: the response $Y_i(1)$
that is realized if $T_i = 1$ (treated), and another response $Y_i(0)$
realized if $T_i =0$ (control). We assume the observational units are
independent and identically distributed samples from a population, for
which we wish to infer the treatment's effect.
The main obstacle is that only one potential outcome is
observed: $Y_i = T_i Y_i(1) - (1-T_i) Y_i(0)$, which is commonly known
as the ``fundamental problem of causal inference" \citep{Holland1986}.

In this paper we focus on estimating the Population Average Treatment effect on the Treated (PATT):
\begin{equation}
\label{eq:patt}
\gamma = \mathrm{E}[Y(1)|T=1] - \mathrm{E}[Y(0)|T=1]
\overset{\Delta}{=} \mu(1|1) - \mu(0|1).
\end{equation}
The counterfactual mean $\mu(0|1) = \mathrm{E}[Y(0)|T=1]$ also naturally occurs in survey sampling with missing
data \citep{deville1992calibration,sarndal2005estimation} by viewing $Y(0)$
as the only outcome of interest (so $T=1$ stands for non-response).

Along with the treatment exposure $T_i$ and outcome $Y_i$, each
unit $i$ is usually associated with a set of
covariates denoted by $X_i$ measured prior to the treatment
assignment. In a typical observational study, both treatment
assignment and outcome may be related to the covariates, which can
cause serious confounding bias. The seminal
work by \citet{rosenbaum1983} suggest that it is possible to correct
the confounding bias under the following two assumptions:

\begin{assumption}{1}{strong ignorability}
  \label{assump:str-ign}
  $(Y(0), Y(1)) \perp T~|~X$.
\end{assumption}
\begin{assumption}{2}{overlap}
  \label{assump:overlap}
  $0 < \mathrm{P}(T=1|X) < 1$.
\end{assumption}

Intuitively, the first assumption says that the observed covariates
contain all the information that may cause the selection bias,
i.e.\ there is no unmeasured confounding variable, and the
second assumption ensures that the bias-correction information is available
across the entire domain of $X$. 

Since the covariates $X$ contain all the information of
confounding bias, it is important to understand the relationship between
$T,Y$ and $X$. Under Assumption \ref{assump:str-ign}, the joint
distribution of $(X,Y,T)$ is determined by the marginal
distribution of $X$ and two conditional distributions given $X$. The first
conditional distribution $e(X) = P(T=1|X)$ is often called the
propensity score and plays a central role in causal
inference \citep{rosenbaum1983}. The second conditional distribution is
the density of $Y(0)$ and $Y(1)$ given $X$. Since we only consider the
mean causal effect in this paper, it suffices to study the mean regression
functions $g_0(X) = E[Y(0)|X]$ and $g_1(X) = E[Y(1)|X]$.

To estimate the PATT defined in \eqref{eq:patt}, a conventional weighting estimator based on the
propensity score is the inverse probability weighting (IPW) defined as
\begin{align} \label{eq:gamma-ipw}
  \hat{\gamma}^{\textrm{IPW}} &= \sum_{T_i = 1} \frac{1}{n_1} Y_i -
  \sum_{T_i = 0} \frac{ \hat{e}(X_i) (1 - \hat{e}(X_i))^{-1}}{\sum_{T_i = 0}
    \hat{e}(X_i) (1 - \hat{e}(X_i))^{-1}} Y_i.
\end{align}
Here $\sum_{T_i = t}$ is a short-hand notation of summation over all units $i$ such
that $T_i = t$. This  will be repeatedly used throughout this paper.
In \eqref{eq:gamma-ipw}, the control units are weighted proportionally
to $\hat{e}(X_i) (1 - \hat{e}(X_i))^{-1}$ to resemble the full
population. The most popular choice of propensity score model is the
logistic regression, where
$\mathrm{logit}(e(x)) = \log [e(x) / (1 - e(x))]$ is modeled by
$\sum_{j=1}^p \theta_j c_j(x)$ and $c_j(x)$ are functions of the
covariates.

\section{Entropy Balancing}
\label{sec:entropy-balancing}

Entropy Balancing (EB) is an alternative weighting method proposed by \citet{Hainmueller2011} to estimate PATT.
EB operates by maximizing the entropy of the weights
under some pre-specified balancing constraints:
\begin{equation}
  \begin{aligned}
  \underset{w}{\mathrm{maximize}} \quad & - \sum_{T_i = 0} w_i \log w_i  \\
  \mathrm{subject~to} \quad & \sum_{T_i = 0} w_i c_j(X_i) =
  \bar{c}_j(1) = \frac{1}{n_1} \sum_{T_i = 1} c_j(X_i), ~ j =
  1, \ldots, p,\\
                            & \sum_{T_i = 0} w_i = 1, \\
                            & w_i > 0,~ i=1,\dotsc,n.
  \end{aligned}
  \label{eq:orig-eb}
\end{equation}
\citet{Hainmueller2011} proposes to use the weighted average
$\sum_{T_i = 0} w_i^\textrm{EB} Y_i$ to estimate the counterfactual
mean $\mathrm{E}[Y(0)|T=1]$. This gives the Entropy Balancing estimator of PATT
\begin{equation}
  \label{eq:eb-patt}
  \hat{\gamma}^{\textrm{EB}} = \sum_{T_i = 1} \frac{Y_i}{n_1} - \sum_{T_i = 0} w_i^{\textrm{EB}}
Y_i.
\end{equation}

The functions $\{c_j(\cdot)\}_{j=1}^p$ in \eqref{eq:orig-eb} are called moment functions of
the covariates. They can be any transformation of $X$, not necessarily polynomial
functions. We use $c(X)$ and $\bar{c}(1)$ to stand for
the vector of $c_j(X)$ and $\bar{c}_j(1)$, $j=1,\ldots,p$. We shall
see the functions $\{c_j(\cdot)\}_{j=1}^p$ indeed serve as the linear predictors in the propensity score
model and the outcome regression model, although at
this point it is not even clear that EB attempts to fit any model.

First, we give some heuristics that allows us to view EB as a propensity score
weighting method. Since EB seeks to empirically match the control and
treatment covariate distributions, we connect EB with density
estimation. 
Let $m(x)$ be the density function of the
covariates $X$ for the control
population. The minimum relative entropy principle estimates the
density of the treatment population by
\begin{equation}
  \label{eq:rel-ent}
  \underset{\tilde{m}}{\mathrm{maximize}}~H(\tilde{m} \| m) \quad
  \mathrm{subject~to} ~ \mathrm{E}_{\tilde{m}}[c(X)] = \bar{c}(1),
\end{equation}
where $H(\tilde{m}\|m) = \mathrm{E}_{\tilde{m}}[\log
(\tilde{m}(X)/m(X))]$ is the relative entropy between $\tilde{m}$ and
$m$.
As an estimate of the
distribution of the treatment group, the optimal $\tilde{m}$
of \eqref{eq:rel-ent} is the ``closest'' to the control
distribution among all distributions satisfying the moment
constraints. Now let $ w(x) = [\mathrm{P}(T=1) \cdot \tilde{m}(x)]/[
  \mathrm{P}(T=0) \cdot m(x)]$ be the population version of the inverse probability weights in
\eqref{eq:gamma-ipw}. Applying a change of measure, we can rewrite
\eqref{eq:rel-ent} as an optimization problem over $w(x)$:
\begin{equation}
  \label{eq:eb-pop}
  \underset{w}{\mathrm{maximize}}~\mathrm{E}_m[w(X) \log w(X)] \quad
  \mathrm{subject~to} ~ \mathrm{E}_{m}[w(X) c(X)] = \bar{c}(1).
\end{equation}
The EB optimization problem \eqref{eq:orig-eb} is the finite
sample version of \eqref{eq:eb-pop}, where the population distribution $m$ is
replaced by the empirical distribution of the control units.

Using the Lagrangian multipliers, one can show the solution
of \eqref{eq:rel-ent} belongs to the family of exponential titled
distributions of $m$ \citep{cover2012}:
\[
m_{\theta}(x) = m(x) \exp(\theta^T c(x) - \psi(\theta)).
\]
Here, $\psi(\theta)$ is the moment generating function of this
exponential family. Consequently, the solution of the population EB \eqref{eq:eb-pop} is
\begin{equation*} \label{eq:logistic-reg}
\frac{e(x)}{1-e(x)} = \frac{\mathrm{P}(T=1|X=x)}{\mathrm{P}(T=0|X=x)} = w(x) = \exp( \alpha
+ \theta^T c(x)),
\end{equation*}
where $\alpha =
\log(\mathrm{P}(T=1)/\mathrm{P}(T=0))$. This is exactly the logistic
regression model with predictors $c(x)$.

Notice that EB is different from the maximum likelihood
fit of the logistic regression. The dual optimization problem of
\eqref{eq:orig-eb} is
\begin{equation}
\label{eq:orig-eb-dual}
\underset{\theta}{\mathrm{minimize}} \quad \log \left(\sum_{T_i=0} \exp\bigg(\sum_{j=1}^p \theta_j
  c_j(X_i)\bigg)\right) - \sum_{j=1}^p \theta_j \bar{c}_j(1),
\end{equation}
whereas the maximum likelihood solves
\begin{equation}
  \label{eq:logistic}
  \underset{\theta}{\mathrm{minimize}} \quad \sum_{i=1}^n \log\left(1 +
  \exp\bigg(-(2 T_i - 1) \sum_{j=1}^p \theta_j c_j(X_i)\bigg) \right).
\end{equation}
It is apparent from \eqref{eq:orig-eb-dual} and \eqref{eq:logistic}
that EB and maximum likelihood use different loss functions. As a
remark, the estimating equations defined by \eqref{eq:orig-eb-dual}
are used to augment the estimating equations defined by
\eqref{eq:logistic} in the covariate balancing propensity score (CBPS)
approach of \citet{Imai2014}. We will compare the empirical
performance of these methods in Section \ref{sec:simulations}.

The optimization problem \eqref{eq:orig-eb-dual} is
strictly convex and the unique solution $\hat{\theta}^{\mathrm{EB}}$
can be efficiently computed by Newton method.
The EB weights (solution to the primal problem \eqref{eq:orig-eb}) are
given by the Karush-Kuhn-Tucker (KKT) conditions: for any $i$ such
that $T_i = 0$,
\begin{equation}
  \label{eq:orig-eb-link}
  w_i^{\textrm{EB}} = \frac{\exp\left(\sum_{j=1}^p \hat{\theta}_j^{\textrm{EB}} c_j(X_i)\right)}{\sum_{T_i=0}\exp\left(\sum_{j=1}^p
    \hat{\theta}_j^{\textrm{EB}} c_p(X_i)\right)}.
\end{equation}

As a final remark, Entropy Balancing bridges two existing approaches of
estimating the mean causal effect:
\begin{enumerate}
\item The calibration estimator that is very popular in survey sampling \citep{deville1992calibration,sarndal2005estimation,chan2015};
\item The empirical likelihood approach that significantly advances
  the theory of doubly robust estimation in observation study \citep{Wang2002,Tan2006,Qin2007,Tan2010}.
\end{enumerate}
EB is a special case of these two approaches. The main distinction is that it uses the Shannon entropy $\sum_{i=1}^n w_i \log w_i$ as the discrepancy
function, resulting in an easy-to-solve convex optimization. Due to
its easy interpretation, Entropy Balancing has already
gained some ground in practice \citep[e.g.][]{Marcus2013,Ferwerda2014}.

\section{Properties of Entropy Balancing}
\label{sec:theoretical-properties}

We give some theoretical guarantees of Entropy Balancing to justify
its usage in real applications. The following is the main theorem of this paper,
which shows EB is doubly robust even though its original form
\eqref{eq:orig-eb} does not contain a propensity score model or a
outcome regression model.

\begin{theorem}
  \label{thm:main}
  Let Assumption \ref{assump:str-ign} and Assumption
  \ref{assump:overlap} be given. Additionally, assume the expectation of
  $c(x)$ exists and $\mathrm{Var}(Y(0)) < \infty$. Then Entropy
  Balancing is doubly robust (Property \ref{property:dr}) in the sense that
  \begin{enumerate}
  \item If $\mathrm{logit}(e(x))$
  or $g_0(x)$ is linear in $c_j(x),~ j=1,\ldots,R$, then
  $\hat{\gamma}^\mathrm{EB}$ is statistically consistent.
  \item Moreover, if $\mathrm{logit}(e(x))$, $g_0(x)$
    and $g_1(x)$ are all linear in $c_j(x),~ j=1,\ldots,R$,
  then $\hat{\gamma}^{\mathrm{EB}}$ reaches the semiparametric
  variance bound of $\gamma$ derived in \citet[Theorem 1]{Hahn1998} with unknown
  propensity score.
  \end{enumerate}
\end{theorem}



We give two proofs of the first claim in Theorem \ref{thm:main}. The
first proof reveals an interesting connection between the
primal-dual optimization problems \eqref{eq:orig-eb} and
\eqref{eq:orig-eb-dual} and the statistical property, double
robustness, which motivates the interpretation in Figure
\ref{fig:ps-or-cb}. The second proof uses a stabilization trick in \citet{Robins2007}.

\begin{sketchproof}[First proof (sketch)]
The consistency under the linear model of
$\mathrm{logit}(\mathrm{P}(T=1|X))$ is a consequence of the dual
optimization problem \eqref{eq:orig-eb-dual}. See Section
\ref{sec:entropy-balancing} for a heuristic justification via the
minimum relative entropy principle and Appendix
\ref{sec:proof-theor-refthm} for a rigorous proof by using the
M-estimation theory.

The consistency under the linear model of $Y(0)$ can be proved by
expanding $\mathrm{E}[Y(0)|X]$ and $\sum_{T_i=0}w_i Y_i$. Here we
provide an indirect proof by showing that augmenting EB with a linear
outcome regression does not change the estimator. Given an estimated
propensity score model $\hat{e}(x)$, the corresponding weights
$\hat{e}(x) / (1 - \hat{e}(x))$ for the control units, and an
estimated outcome regression model $\hat{g}_0(x)$, a doubly robust
estimator of PATT is given by
\begin{equation}
  \label{eq:eb-dr}
  \hat{\gamma}^{\textrm{DR}} = \sum\limits_{T_i = 1} \frac{1}{n_1}
  (Y_i - \hat{g}_0(X_i)) - \sum\limits_{T_i = 0} \frac{\hat{e}(X_i)}{1
  - \hat{e}(X_i)} (Y_i - \hat{g}_0(X_i)).
\end{equation}
This estimator satisfies Property \ref{property:dr}, i.e.\ if
$\hat{e}(x) \to e(x)$ or $\hat{g}_0(x) \to g(x)$, then
$\hat{\gamma}^{\textrm{DR}}$ is statistically consistent for
$\gamma$. To see this, in the case that $\hat{g}_0(x)
\to g_0(x)$, the first sum in \eqref{eq:eb-dr} is consistent
for $\gamma$ and the second sum in \eqref{eq:eb-dr} has mean going to
$0$ as $n \to \infty$. In the case where $\hat{g}_0(x)
\not \to g_0(x)$ but $\hat{e}(x) \to e(x)$, the second sum in
\eqref{eq:eb-dr} is consistent for the bias of the first sum (as an
estimator of $\gamma$).

When the estimated propensity score model $\hat{e}(x)$ is obtained by
the EB dual problem \eqref{eq:orig-eb-dual} and the estimated outcome
regression model is $\hat{g}_0(x) = \sum_{j=1}^p \hat{\beta}_j
c_j(x)$, we have
\[
\begin{aligned}
\hat{\gamma}^{\textrm{DR}}
  - \hat{\gamma}^{\textrm{EB}}
&= \sum_{T_i=0} w_i^{\textrm{EB}} \hat{g}_0(X_i) - \frac{1}{n_1} \sum_{T_i=0} \hat{g}_0(X_i) \\
 &= \sum_{T_i = 0} w_i^{\textrm{EB}} \sum_{j=1}^{p} \hat{\beta}_j c_j(X_i)
  - \frac{1}{n_1} \sum_{T_i = 1} \sum_{j=1}^p\hat{\beta}_j c_j(X_i) \\
  &=\sum_{j=1}^p \hat{\beta}_j\left( \sum_{T_i = 0} w_i^{\textrm{EB}}
    c_j(X_i) - \frac{1}{n_1} \sum_{T_i=1} c_j(X_i) \right) \\
  &= 0.
\end{aligned}
\]
Therefore by enforcing covariate balancing constraints, EB implicitly
fits a linear outcome regression model and is consistent for $\gamma$
under this model.
\end{sketchproof}

\begin{sketchproof}[Second proof]
This proof is pointed out by an anonymous reviewer. In a
discussion of \citet{Kang2007}, \citet{Robins2007} show that one
can stabilize the standard doubly robust estimator in a number of
ways. Specifically, one trick suggested by \citet[Section 4.1.2]{Robins2007} is to
estimate the propensity score, say $\tilde{e}(x)$, by the following estimating equation
\begin{equation}
  \label{eq:eb-est-eq}
  \sum_{i=1}^n \left[ \frac{(1-T_i) \tilde{e}(X_i) / (1 -
      \tilde{e}(X_i))}{\sum_{i=1}^n (1-T_i) \tilde{e}(X_i) / (1
        - \tilde{e}(X_i))} - \frac{T_i}{\sum_{i=1}^n T_i} \right]
  \hat{g}_0(X_i) = 0.
\end{equation}
Then one can estimate PATT by the IPW estimator \eqref{eq:gamma-ipw}
by replacing $\hat{e}(X_i)$ with $\tilde{e}(X_i)$. This estimator is
sample bounded (the estimator is always within the range of observed
values of $Y$) and doubly robust with respect to the parametric
specifications of $\tilde{e}(x) = \tilde{e}(x;\theta)$ and
$\hat{g}_0(x) = \hat{g}_0(x;\beta)$. The only problem with
\eqref{eq:eb-est-eq} is it may not have a unique solution. However,
when $\mathrm{logit}(e(x))$ and $g_0(x)$ are assumed linear in $c(x)$,
\eqref{eq:eb-est-eq} corresponds to the first order condition of the
EB dual problem \eqref{eq:orig-eb-dual}. Since \eqref{eq:orig-eb-dual}
is strictly convex, it has an unique solution and
$\tilde{e}(X;\theta)$ is the same as the EB estimate
$\hat{e}(X;\theta)$. As a consequence, $\hat{\gamma}^{\textrm{EB}}$ is
also doubly robust.
\end{sketchproof}

To prove the second claim in Theorem \ref{thm:main}, we compute the
asymptotic variance of $\hat{\gamma}^{EB}$ using the M-estimation theory. 
To state our
results, we need to introduce four differently weighted
covariance-like functions for two random vectors
$a_1$ and $a_2$ of length $p$:
\begin{equation*}
\begin{aligned}
H_{a_1,a_2} &= \mathrm{Cov}(a_1, a_2|T=1), \\
G_{a_1,a_2} &= \mathrm{E}\left[ \frac{e(X)}{1-e(X)} (a_1 - \mathrm{E}[a_1|T=1])
  (a_2 - \mathrm{E}[a_2|T=1])^T\middle|T=1\right], \\
K_{a_1, a_2} &= \mathrm{E}[(1-e(X)) a_1 a_2^T | T = 1], \\
K^m_{a_1, a_2} &= \mathrm{E}[(1-e(X)) a_1 (a_2 - \mathrm{E}[a_2|T=1]) ^T | T = 1].
\end{aligned}
\end{equation*}

It is obvious that $H \ge
K$ and usually $G \ge H$. To make the notation more concise, $c(X)$
will be abbreviated as $c$ and $Y(0)$ as $0$ in subscripts. For
example, $H_{c,0} = H_{c(X),Y(0)}$, $G_{c,1} = G_{c(X),Y(1)}$ and $K_c = K_{c(X),c(X)}$.

\begin{theorem}
  \label{thm:variance}
  Assume the logistic regression model of propensity score is correct,
  i.e.\ $\mathrm{logit}(\mathrm{P}(T=1|X))$ is a linear combination of
  $\{c_j(X)\}_{j=1}^p$. Let $\pi = \mathrm{P}(T=1)$, then we have $\hat{\gamma}^{\mathrm{EB}}
  \overset{d}{\to} \mathrm{N}(\gamma,V^{\mathrm{EB}}/n)$ and $\hat{\gamma}^{\mathrm{IPW}}
  \overset{d}{\to} \mathrm{N}(\gamma,V^{\mathrm{IPW}}/n)$ where
\begin{align}
\label{eq:eb-var}
V^{\mathrm{EB}} &= \pi^{-1} \cdot \left\{ H_{1} +  G_{0} -
  H_{c,0}^T H_{c}^{-1} \left( 2 G_{c,0} - H_{c,0} - G_c
  H_{c}^{-1} H_{c,0} + 2 H_{c,1}\right)  \right\}, \\
\label{eq:ipw-var}
  V^{\mathrm{IPW}} &= \pi^{-1} \cdot \left\{H_{1} + G_{0} -
H_{c,0}^T K_{c}^{-1} \left( H_{c,0} - 2 K^m_{c,0} + 2 K^m_{c,1} \right)
\right\}.
\end{align}
\end{theorem}

The proof of Theorem \ref{thm:variance} is given in Appendix
\ref{sec:proof-theor-refthm}. The $H$, $G$ and $K$ matrices in Theorem \ref{thm:variance} can be
estimated from the observed data, yielding approximate sampling
variances for $\hat{\gamma}^{\textrm{EB}}$ and
$\hat{\gamma}^{\text{IPW}}$. Alternatively, variance estimates may be
obtained via the empirical sandwich method \citep[e.g.][]{Stefanski2002}. In practice (particularly in
simulations where we compare to a known truth), we find that the
empirical sandwich method is more stable than the plug-in method,
which is consistent with the suggestion in \citet{Lunceford2004} for PATE estimators.

To complete the proof of the second claim in Theorem \ref{thm:main},
we compare these variances with the semiparametric variance bound of
$\gamma$ with unknown $e(X)$ derived by \citet[Theorem 1]{Hahn1998}:
\begin{equation*}
V^{*}
= \frac{1}{\pi^2} \mathrm{E} \left[
e(X) \mathrm{Var}(Y(1)|X) + \frac{e(X)^{2}}{1 - e(X)} \mathrm{Var}(Y(0)|X) + e(X)
(g_1(X) - g_0(X) - \gamma)^2
\right]
\end{equation*}
After some algebra, one can express $V^{*}$ in terms of
$H_{\cdot,\cdot}$ and $G_{\cdot,\cdot}$ defined above:
\[
V^{*}= \pi^{-1} \cdot \left\{H_1 + G_0 - 2H_{g_0,g_1} - G_{g_0} + H_{g_0} \right\}.
\]
Now assume
$\mathrm{logit}(\mathrm{P}(T=1|X)) = \theta^T c(X)$ and
$\mathrm{E}[Y(t)|X] = \beta(t)^T c(X),~t=0,1$, it is easy to verify that
\[
\begin{aligned}
H_{c,t} &= \mathrm{Cov}(c(X),Y(t)|T=1) \\
&= \mathrm{Cov}(c(X), \beta(t)^T c(X)|T=1) \\
&= H_c \beta(t),~\mathrm{for}~t=0,1.
\end{aligned}
\]
Similarly, $G_{c,t} = G_c \beta(t),~t=0,1$. From here it is easy to
check $V^{\textrm{EB}}$ and $V^{*}$ are the same. Since Entropy
Balancing reaches the efficiency bound in this case, obviously
$V^{\textrm{EB}} < V^{\textrm{IPW}}$ when both models are linear.

If $\mathrm{logit}(\mathrm{P}(T=1|X)) = \theta^T c(X)$ is true
but $\mathrm{E}[Y(t)|X] = \beta(t)^T c(X)$ is not true for some $t =
0,1$, there is no guarantee that EB has the smaller asymptotic
variance. In practice, the features $c(X)$ in the
outcome regression models are almost always correlated with $Y$. This correlation
compensates the slight efficiency loss of not maximizing the likelihood
function in logistic regression. As a consequence, the variance
$V^{\textrm{EB}}$ in \eqref{eq:eb-var} is usually smaller than
$V^{\textrm{IPW}}$ in \eqref{eq:ipw-var}. This efficiency advantage of
EB over IPW is verified in the next section using simulations.

\section{Simulations}
\label{sec:simulations}

\subsection{Kang-Schafer Example}
\label{sec:kang-schafer-example}

We use the simulation example in \citet{Kang2007} to compare EB
weighting with IPW (after maximum likelihood logistic regression) and
the over-identified Covariate Balancing
Propensity Score (CBPS) proposed by \citet{Imai2014}. The simulated data consist of $\{X_i, Z_i, T_i, Y_i\}, i=1,\ldots,n \}$. $X_i$ and $T_i$ are always observed, $Y_i$ is observed only if $T_i = 1$, and $Z_i$ is never observed. To generate this data set, $X_i$ is distributed as $\mathrm{N}(0,I_4)$, $Z_i$ is computed by first applying the following  transformation:
\begin{align*}
Z_{i1} &= \exp(X_{i1}/2), \\
Z_{i2} &= X_{i2}/(1 + \exp(X_{i1})) + 10, \\
Z_{i3} &= (X_{i1}X_{i3} + 0.6)^3, \\
Z_{i4} &= (X_{i2} + X_{i4} + 20)^2.
\end{align*}
Next we normalize each column such that $Z_i$ has mean $0$ and standard deviation $1$.

In one setting, $Y_i$ is generated by $Y_i = 210 + 27.4 X_{i1} + 13.7 X_{i2} + 13.7 X_{i3} + 13.7 X_{i4} + \epsilon_i$, $\epsilon_i \sim \mathrm{N}(0,1)$ and the true propensity scores are $e_i = \mathrm{expit}(-X_{i1} + 0.5 X_{i2} - 0.25 X_{i3} -0.1 X_{i4})$. In this case, both $Y$ and $T$ can be correctly modeled by (generalized) linear model of the observed covariates $X$.

In the other settings, at least one of the propensity score model and the outcome regression model is incorrect. In order to achieve this, the data generating process described above is altered such that $Y$ or $T$ (or both) is linear in the unobserved $Z$ instead of the observed $X$, though the parameters are kept the same.

For each setting ($4$ in total), we generated $1000$ simulated data
sets of size $n=200$ and $1000$, then apply various methods discussed
earlier including
\begin{enumerate}
\item{IPW, CBPS}: the IPW estimator in \eqref{eq:gamma-ipw} with propensity
  score estimated by logistic regression or CBPS (since the estimand is
  overall mean, we use the CBPS weights tailored for estimating PATE);
\item{EB}: the Entropy Balancing estimator (the EB weights are used to
  estimate the unobserved mean $\mathrm{E}[Y|T = 0]$);
\item{IPW+DR, CBPS+DR}: the doubly robust estimator in
  \eqref{eq:eb-dr} with propensity score estimated by logistic
  regression or CBPS.
\end{enumerate}

\begin{figure}[htbp]
  \centering
  \begin{subfigure}[b]{0.8\textwidth}
  \includegraphics[width = \textwidth]{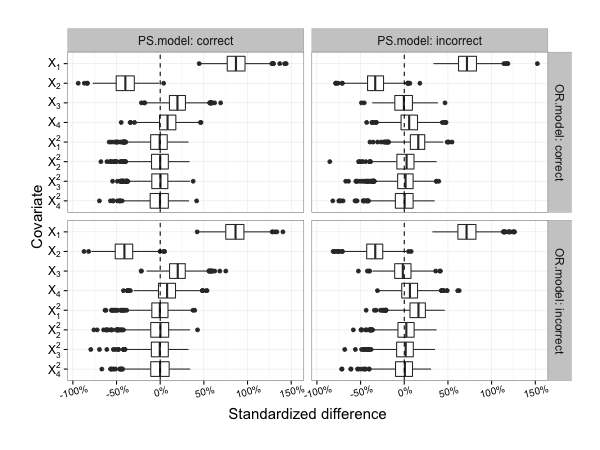}
  \caption{Covariate imbalance before adjustment.}
  \label{fig:KS-imba-n200}
  \end{subfigure}
  \begin{subfigure}[b]{0.8\textwidth}
  \includegraphics[width = \textwidth]{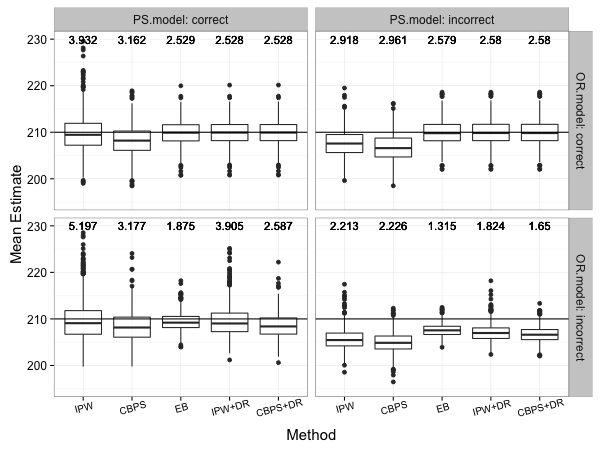}
  \caption{Mean estimates. The methods are: Inverse Propensity Weighting (IPW), Covariate Balancing Propensity Score (CBPS), Entropy Balancing (EB), and doubly robust versions of the first two (IPW+DR, CBPS+DR). Target mean is $210$ and is marked as a black
  horizontal line to compare the biases. Numbers
  printed at $Y=230$ are the sample standard deviations to compare efficiency.}
  \label{fig:KS-mean-n200}
  \end{subfigure}
  \caption{Kang-Schafer example: sample size $n = 200$. Both propensity score
    model and outcome regression model can be correct or incorrect, so
    there are four scenarios in total. We generate 1000 simulations in
  each scenario.}
  \label{fig:KS-n200}
\end{figure}

\begin{figure}[htbp]
  \centering
  \begin{subfigure}[b]{0.8\textwidth}
  \includegraphics[width = \textwidth]{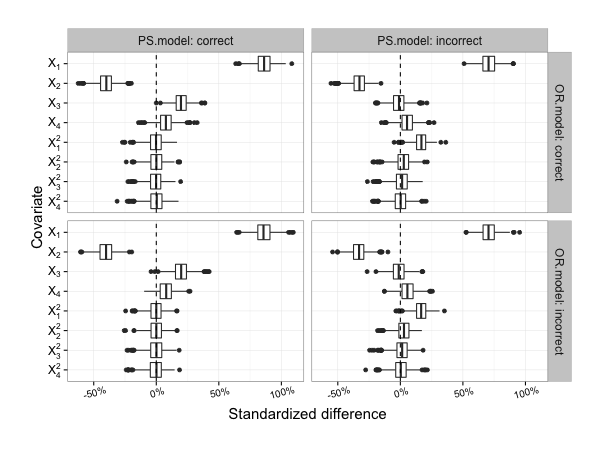}
  \caption{Covariate imbalance before adjustment.}
  \label{fig:KS-imba-n1000}
  \end{subfigure}
  \begin{subfigure}[b]{0.8\textwidth}
  \includegraphics[width = \textwidth]{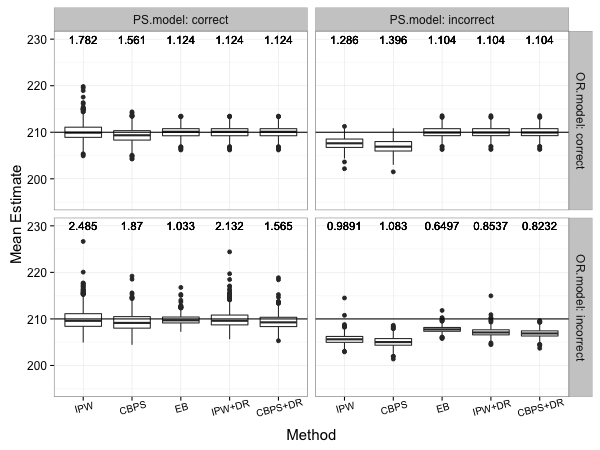}
  \caption{Mean estimates. The methods are: Inverse Propensity Weighting (IPW), Covariate Balancing Propensity Score (CBPS), Entropy Balancing (EB), and doubly robust versions of the first two (IPW+DR, CBPS+DR). Target mean is $210$ and is marked as a black
  horizontal line to compare the biases. Numbers
  printed at $Y=230$ are the sample standard deviations to compare efficiency.}
  \label{fig:KS-mean-n1000}
  \end{subfigure}
  \caption{Kang-Schafer example: sample size $n = 1000$. Both propensity score
    model and outcome regression model can be correct or incorrect, so
    there are four scenarios in total. We generate 1000 simulations in
  each scenario.}
  \label{fig:KS-n1000}
\end{figure}

The simulation results are presented in Figures \ref{fig:KS-n200} and
\ref{fig:KS-n1000}. Figures \ref{fig:KS-imba-n200} and
\ref{fig:KS-imba-n1000} show the covariate imbalance before adjustment
in terms of standardized difference. Figures \ref{fig:KS-mean-n200}
and \ref{fig:KS-mean-n1000} show the mean estimates given by the five
different methods. First, notice that the doubly robust estimator
``IPW+DR" performs poorly when both models are misspecified
(bottom-right panel in Figures \ref{fig:KS-mean-n200} and
\ref{fig:KS-mean-n1000}). In fact,
all the three doubly robust methods are worse than just using
IPW. Second, the three doubly robust estimators have exactly the same
variance if the $Y$ model is correct (top two panels in Figures
\ref{fig:KS-mean-n200} and \ref{fig:KS-mean-n1000}). It seems that how
one fits the propensity score model
has no impact on the final estimate. This is related to the
observation in \citet{Kang2007} that, in this example, the plain OLS
estimate of $Y$ actually outperforms any method involving the
propensity score model. Discussion articles such as \citet{Robins2007} and
\cite{Ridgeway2007} find this phenomenon very uncommon in practice and
is most likely due to the estimated inverse probability weights are
highly variable, which is a bad setting for doubly robust estimators.

Regarding Entropy Balancing (EB), we find that:
\begin{enumerate}
\item If both $T$ and $Y$ models are misspecified, EB
  has smaller bias than the conventional ``IPW+DR" or ``CBPS+DR". So
  EB seems to be less affected by such
  unfavorable setting.
\item When $T$ model is correct but $Y$ model is wrong (bottom-left
  panel in Figures \ref{fig:KS-n200} and \ref{fig:KS-n1000}), EB has the smallest
  variance among all estimators. This supports the conclusion of our
  efficiency comparison of IPW and EB in Section
  \ref{sec:theoretical-properties}.
\end{enumerate}

Finally notice that the same simulation setting is used in
\citet{Tan2010} to study the performance of a number of
doubly robust estimators. The reader can
compare the Figures \ref{fig:KS-n200} and \ref{fig:KS-n1000} with the
results there. The performance of Entropy Balancing is comparable to
the best estimator in \citet{Tan2010}.

\subsection{Lunceford-Davidian Example}
\label{sec:lunc-david-example}

We provide another simulation example by
\citet{Lunceford2004} to verify claims in Theorems \ref{thm:main} and
\ref{thm:variance}. In this simulation, the data still consist of
$\{(X_i, Z_i, T_i, Y_i), i=1,\ldots,n \}$, but all of them are
observed. Both $X_i$ and $Z_i$ are three dimensional vectors. The
propensity score is only related to $X$ through:

$$\mathrm{logit}(\mathrm{P}(T_i=1)) = \beta_0 + \sum_{j=1} \beta_j
X_{ij}.$$
Note the above does not involve elements of $Z_i$. The response $Y$ is
generated according to
$$Y_i = \nu_0 + \sum_{j=1}^{3} \nu_j X_{ij} +
\nu_4 T_i +
\sum_{j=1}^3 \xi_j Z_{ij} + \epsilon_i; \ \epsilon_i \sim
\mathrm{N}(0,1).$$ The parameters here are set to be
$\nu = (0, -1, 1, -1, 2)^T$, and $\beta$ is set as:
\begin{align*}
\beta^{\textrm{no}} &= (0, 0, 0, 0)^T, \\
\beta^{\textrm{moderate}} &= (0, 0.3, -0.3, 0.3)^T, \  \mbox{or}\\
\beta^{\textrm{strong}} &= (0, 0.6, -0.6, 0.6)^T.
\end{align*}
The choice of $\beta$ depends on the
level of association of $T$ and $X$. $\xi$ is based on a similar choice on the level of association of $Y$ and $Z$:
\begin{align*}
\xi^{\textrm{no}} &= (0, 0, 0)^{T},  \\
\xi^{\textrm{moderate}} &= (-0.5, 0.5, 0.5)^{T}, \ \mbox{or} \\
\xi^{\textrm{strong}} &= (-1, 1, 1)^{T}.
\end{align*}
The joint distribution of $(X_i, Z_i)$ is specified by taking $X_{i3} \sim
\mathrm{Bernoulli}(0.2)$ and then generate $Z_{i3}$ as Bernoulli with
$$\mathrm{P}\left(Z_{i3}=1|X_{i3}\right) = 0.75 X_{i3} + 0.25(1 -
X_{i3}).$$
Conditional on $X_{i3}$, $(X_{i1},Z_{i1},X_{i2},Z_{i2})$ is
then generated as multivariate normal
$\mathrm{N}(a_{X_{i3}},B_{X_{i3}})$, where $a_1 = (1, 1, -1, -1)^T$,
$a_0 = (-1, -1, 1, 1)^{T}$ and
$$B_0 = B_1 =
\begin{pmatrix}
  1 & 0.5 & -0.5 & -0.5 \\
  0.5 & 1 & -0.5 & -0.5 \\
 -0.5 & -0.5 & 1 & 0.5 \\
 -0.5 & -0.5 & 0.5 & 1 \\
\end{pmatrix}.
$$
Figure \ref{fig:ld-imba} shows the covariate imbalance in the three
settings before adjustment.

\begin{figure}[ht]
  \centering
    \includegraphics[width = 0.8\textwidth]{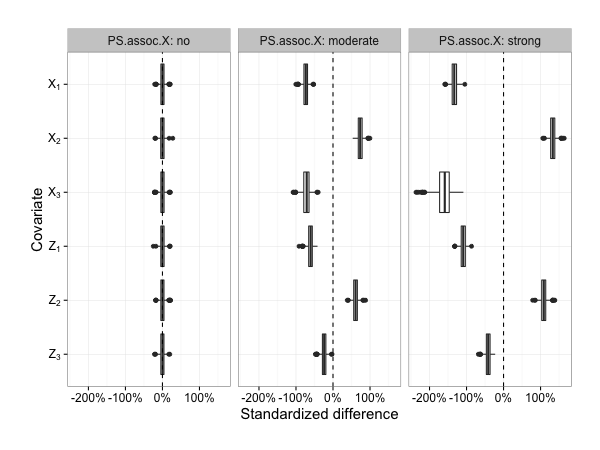}
  \caption{Lunceford-Davidian Example: covariate imbalance before adjustment.}
  \label{fig:ld-imba}
\end{figure}

The data generating model implies that the true PATT is $\gamma = 2$. Since the
outcome $Y$ depends on both $X$ and $Z$, we always fit a full linear
model of $Y$ using $X$ and $Z$, if such model is needed. $T$
only depends on $X$, so it is not necessary to include $Z$ in
propensity score modeling. However, as pointed out by
\citet[Sec. 3.3]{Lunceford2004}, it is actually beneficial to
``overmodel" the propensity score by including $Z$ in the model. Here
we will try both possibilities, the ``full" modeling of $T$ using both
$X$ and $Z$, and the ``partial" modeling of $T$ using only $X$. Since
the estimand is PATT in this case, we use the over-identified CBPS
weights tailored for estimating PATT.

\begin{figure}[ht]
  \centering
  \includegraphics[width = \textwidth]{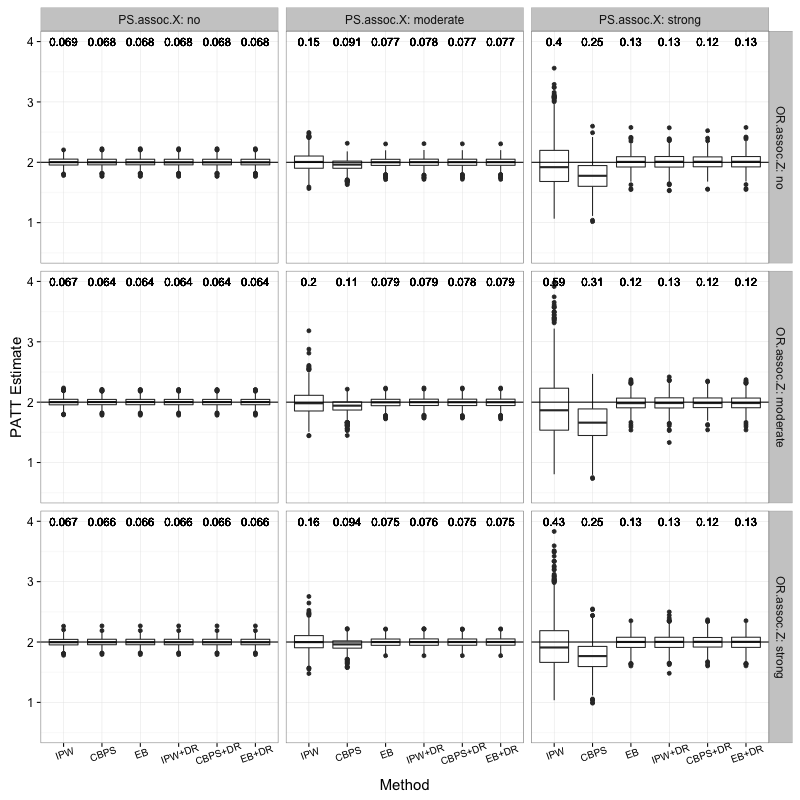}
  \caption{Results of the Lunceford-Davidian example (full
    propensity score modeling). The propensity score model and outcome
  regression model, if applies, are always correctly specified, but
  the level of association between $T$ or $Y$ with $X$ or $Z$ could be
 different, ended up with $9$ different scenarios. $X$ are confounding
 covariates and $Z$ only affects the outcome. We generate $1000$
 simulations of $1000$ in each scenario and apply five
  different estimators. The true PATT is $2$ and is marked as a black
  horizontal line to compare the biases of the methods. Numbers
  printed at $Y=5$ are the sample standard deviation of each method,
in order to compare their efficiency. }
  \label{fig:LD-full}
\end{figure}

\begin{figure}[htbp]
  \centering
  \includegraphics[width = \textwidth]{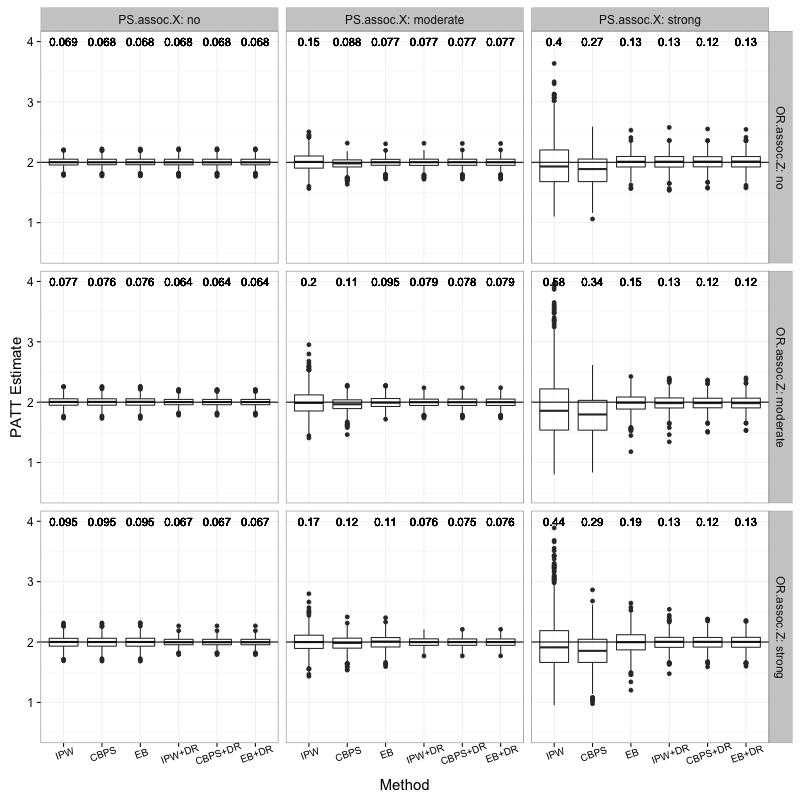}
  \caption{Results of the Lunceford-Davidian example (partial propensity
    score modeling). The settings are exactly the same as Figure
    \ref{fig:LD-full} except the methods here don't use $Z$ in their
    propensity score models.}
  \label{fig:LD-partial}
\end{figure}

We generated $1000$ simulated data sets and the results are shown in Figure \ref{fig:LD-full} for ``full"
propensity score modeling and Figure \ref{fig:LD-partial} for
``partial" propensity score modeling. We make the following
comments about these two plots:
\begin{enumerate}
\item IPW and all other estimators are always
  consistent, no matter what level of association is specified. This
  is because the propensity score model is always correctly specified.
\item When using the ``full'' propensity score modeling, all doubly robust
  estimators (EB, IPW+DR, CBPS+DR and EB+DR) have almost the same sample
  variance. This is because all of them are asymptotically efficient.
\item CBPS, to our surprise, does not perform very well in this
  simulation. It has smaller variance than IPW but this comes with the price of
  some bias. If we use the partial propensity
  score model (only involve $X$, Figure \ref{fig:LD-partial}), this
  bias is smaller but still not negligible. While it is not clear
  what causes this bias, one possible reason is that the
  optimization problem of CBPS is nonconvex, so the local solution
  which is used to construct $\gamma$ estimator could be far from the
  global solution. Another possibility is that CBPS uses GMM or
  Empirical Likelihood to combine likelihood with imbalance penalty,
  which is less efficient than maximum likelihood directly. Thus, although
  the estimator is asymptotically unbiased, the convergence spend to
  the true $\gamma$ is quite slower than IPW. CBPS combined with
  outcome regression (CB+DR) fixes the bias and inefficiency issue occurred in CBPS without outcome regression.
\item EB, in contrast, performs quite well in this
  simulation. It has relatively small variance, particularly if we
  use the ``full" model in which\ both $X$ and $Z$ are balanced.
\item The difference between EB and EB+DR is that while EB only
  balances ``partial" or ``full" covariates, EB+DR additionally
  combines a outcome linear regression model on all the covariates. As
  shown in the first proof of Theorem \ref{thm:main}, when the ``full" covariates are used, EB is exactly the same as EB+DR. We can observe this from Figure \ref{fig:LD-full}. When EB only balances ``partial" covariates, the two methods are different and indeed EB+DR is more efficient in Figure \ref{fig:LD-partial} since it fits the correct $Y$ model.
\item Using the ``full" propensity score model improves the efficiency
  of pure weighting estimators (IPW, CBPS and EB) a lot, but has very
  little impact on estimators that involves an outcome regression model
  (IPW+DR and CBPS+DR) compared to "partial" propensity score modeling. Although EB could be viewed as fitting a
  outcome model implicitly,
  the "partial" EB estimator only uses $X$ in the outcome model,
  that is precisely the reason why it is not efficient. Thus there are
  both robustness and efficiency reasons that one should include all
  relevant covariates in EB, even if the covariates
  affect only one of $T$ and $Y$.
\end{enumerate}

In summary, EB outperforms IPW in all the simulations, making it an
appealing alternative to the conventional propensity score weighting
methods.

\appendix

\section{Theoretical proofs}
\label{sec:proof-theor-refthm}

We first describe
the conditions under which the EB problem \eqref{eq:orig-eb} admits a
solution. The existence of $w^{\textrm{EB}}$ depends on the
solvability of the moment matching constraints
\begin{equation}
  \label{eq:moment-constraints}
   \sum_{T_i = 0} w_i c_j(X_i) = \bar{c}_j(1), ~ j =
  1, \ldots, p,~  w > 0, ~\sum_{T_i=0} w_i = 1.
\end{equation}
As one may expect, this is closely related to the existence condition of maximum likelihood estimate of logistic regression \citep{Silvapulle1981,Albert1984}.
An easy way to obtain such condition is through the dual problem of \eqref{eq:logistic}
\begin{equation}
  \label{eq:logistic-dual}
  \begin{aligned}
  \underset{w}{\mathrm{maximize}} \quad & - \sum_{i=1}^n \left[w_i \log w_i + (1 - w_i)
  \log (1 - w_i) \right]  \\
  \mathrm{subject~to} \quad &
  \sum_{T_i = 0} w_i c_j(X_i) = \sum_{T_i = 1} w_i c_j(X_i),~j =
  1,\ldots,p, \\
  & 0 < w_i < 1,~i=1,\dotsc,n.
  \end{aligned}
\end{equation}
Thus, the existence of $\hat{\theta}^{\textrm{MLE}}$ is equivalent to the
solvability of the constraints in \eqref{eq:logistic-dual}, which is the
overlap condition first given by \citet{Silvapulle1981}. 

Intuitively, in the space of $c(X)$, the solvability of
\eqref{eq:moment-constraints} or the existence of $w^{\textrm{EB}}$ means there is no hyperplane
separating $\{c(X_i)\}_{T_i = 0}$ and $\bar{c}(1)$. In contrast, the solvability of
\eqref{eq:logistic-dual} or the existence of $w^{\textrm{MLE}}$ means there is no hyperplane separating
$\{c(X_i)\}_{T_i = 0}$ and $\{c(X_i)\}_{T_i = 1}$. Hence the existence
of EB requires a stronger condition than the logistic
regression MLE.

The next proposition suggests that the existence of
$w^{\mathrm{EB}}$ and hence $w^{\mathrm{MLE}}$ is guaranteed by Assumption
\ref{assump:overlap} with high probability.
\begin{proposition}
  \label{prop:exist-whp}
  Suppose Assumption \ref{assump:overlap} is satisfied and the expectation of $c(X)$ exist, then $\mathrm{P}(w^{\mathrm{EB}}~\mathrm{exists})
  \to 1$ as $n \to \infty$. Furthermore, $\sum_{i=1}^n
  (w_i^{\mathrm{EB}})^2 \to 0$ in probability as $n \to \infty$.
\end{proposition}
\begin{proof}
Since the expectation of $c(X)$ exist, the weak law of large number
says
$\bar{c}(1) \overset{\mathrm{p}}{\to} \bar{c}^{*}(1) = \mathrm{E}[c(X)|T=1]$. Therefore
\begin{lemma}
\label{lem:c-converge}
For any $\epsilon > 0$, $\mathrm{P}(\|\bar{c}(1) - \bar{c}^{*}(1)\|_{\infty} \ge \epsilon)
\to 0$ as $n \to \infty$.
\end{lemma}

Now condition on $\|\bar{c}(1) - \bar{c}^{*}(1)\|_{\infty} \ge
\epsilon$, i.e.\ $\bar{c}(1)$ is in the box of side length $2\epsilon$
centered at $\bar{c}^{*}(1)$, we want to prove that with probability going
to $1$ there exists $w$ such that $w_i > 0$, $\sum_{T_i = 0} w_i = 1$
and $\sum_{T_i = 0} w_i c(X_i) = \bar{c}(1)$. Equivalently, this is
saying the convex hull generated by $\{c(X_i)\}_{T_i = 0}$ contains
$\bar{c}(1)$. We indeed prove a stronger result:
\begin{lemma}
\label{lem:hull-box}
With probability going to $1$ the convex hull generated by
$\{c(X_i)\}_{T_i = 0}$ contains the box $B_{\epsilon}(\bar{c}^{*}(1)) = \{
c(x): \|c(x) - \bar{c}^{*}(1) \|_{\infty} \le \epsilon \}$ for some
$\epsilon > 0$.
\end{lemma}

Proposition \ref{prop:exist-whp} follow immediately from Lemma
\ref{lem:c-converge} and Lemma \ref{lem:hull-box}. Now we prove Lemma
\ref{lem:hull-box}. Denote the sample space of $X$ by
$\Omega(X)$. Assumption \ref{assump:overlap} implies $\bar{c}^{*}(1)$
hence $B_{\epsilon}(\bar{c}^{*}(1))$ is in the interior of the convex
hull of $\Omega(X)$ for sufficiently small $\epsilon$. Let
$R_i,~i=1,\ldots,3^p$, be the $3^p$ boxes centered at $\bar{c}^{*}(1)
+\frac{3}{2} \epsilon b$, where $b \in \mathbb{R}^p$ is a vector that
each entry can be $-1$, $0$, or $1$. It is easy to check that the sets
$R_i$ are disjoint and the convex hull of $\{x_i\}_{i=1}^{3^p}$ contains
$B_{\epsilon}(\bar{c}^{*}(1))$ if $x_i \in R_i,~i=1,\ldots,3^p$. Since
$0 < P(T=0|X) < 1$, $\rho = \min_i \mathrm{P}(X \in R_i | T = 0) >
0$. This implies
\begin{equation}
\label{eq:whp}
\begin{aligned}
\mathrm{P}(\exists X_i\in R_i~\mathrm{and}~T_i=0,~\forall i = 1, \ldots, 3^p) &\ge 1 -
\sum_{i=1}^{3^p} \mathrm{P}(X \not\in R_i | T=0)^n \\
  &\ge 1 - 3^p (1 - \rho)^n \\
  &\to 1
\end{aligned}
\end{equation}
as $n \to \infty$. This proves the lemma because the event in the left
hand side implies the convex hull generated by $\{c(X_i)\}_{T_i=0}$
contains the desired box. Note that \eqref{eq:whp} also tells us how
many samples we actually need to ensure the existence of
$w^{\textrm{EB}}$. Indeed if $n \ge \rho^{-1}(p
\log 2 + \log \delta^{-1}) \ge \log_{(1 - \rho)} (\delta 2^{-p})$,
then the probability in \eqref{eq:whp} is greater than $1 -
\delta$. Usually we expect $\delta = O(3^{-p})$. If this is the case,
the number of samples needed is $n = O(p \cdot 3^p)$\footnote{Note that this
naive rate can actually be greatly improved by Wendel's theorem in geometric
probability theory.}.

Now we turn to the second claim of the proposition, i.e.\ $\sum_{T_i=0}
w_i^2 \overset{p}{\to} 0$. To prove this, we only need to find a
sequence (with respect to
growing $n$) of feasible solutions to \eqref{eq:orig-eb} such that
$\max_i w_i \to 0$. This is not hard to show, because the probability
in \eqref{eq:whp} is exponentially decaying as $n$ increases. We can
pick $n_1 \ge N(\delta, p, \rho)$ such that the probability of the
convex hull of $\{x_i\}_{i=1}^{n_1}$ contains
$B_{\epsilon}(\bar{c}^{*}(1))$ is at least $1 - \delta$, then pick
$n_{i+1} \ge n_i +3^i N(\delta, p, \rho)$ so the convex hull of $\{x_i\}_{i=n_i+1}^{n_{i+1}}$ contains
$B_{\epsilon}(\bar{c}^{*}(1))$ with probability at least $1 -
3^i\delta$. This means for each
$\{x_i\}_{i=n_i+1}^{n_{i+1}},~i=0,1,\ldots$, we have a set of weights
$\{\tilde{w}_i\}_{i=n_i+1}^{n_{i+1}}$ such that
$\sum_{i=n_i+1}^{n_{i+1}} \tilde{w}_i x_i = \bar{c}(1)$. Now suppose $
n_k \le n < n_{k+1}$, the choice $w_i = \tilde{w}_i / k$ if $i \le
n_k$ and $w_i = 0$ if $i > n_k$ satisfies the constraints and $\max_i
w_i \le k$. As $n \to \infty$, this implies $\max_i w_i \to 0$ and
hence $\sum_i w_i^2 \to 0$ with probability tending to $1$.
\end{proof}


Now we turn to the main theorem of the paper (Theorem \ref{thm:main}).
The first claim in Theorem \ref{thm:main} follows immediately from the following lemma:
\begin{lemma}
\label{lem:robust-t}
  Under the assumptions in Theorem \ref{thm:main} and suppose $\mathrm{logit}(P[T=1|X]) = \sum_{j=1}^p \theta_j^{*} c_j(X)$,
  then as $n \to \infty$, $\hat{\theta}^{\mathrm{EB}} \overset{p}{\to}
  \theta^{*}$. As a consequence,
\[
  \mathrm{E}\left[\sum_{T_i=0} w_i^{\mathrm{EB}} Y_i\right] \overset{p}{\to} \mathrm{E}[Y(0)|T=1].
\]
\end{lemma}
\begin{proof}
  The proof is a standard application of M-estimation (more precisely
  Z-estimation) theory. We will follow the estimating equations
  approach described in \citep{Stefanski2002} to derive consistency of
  $\hat{\theta}^{\textrm{EB}}$.
  First we note the first order optimality
  condition of \eqref{eq:orig-eb-dual} is
  \begin{equation}
    \label{eq:orig-eb-dual-derivative}
    \sum_{i=1}^n (1 - T_i) e^{\sum_{k=1}^p \theta_k c_k(X_i)} (c_j(X_i) - \bar{c}_j(1)) = 0,~j=1,\ldots,R.
  \end{equation}
  We can rewrite \eqref{eq:orig-eb-dual-derivative} as estimating
  equations. Let $\phi_j(X,T;m) = T(c_j(X) - m_j),~j=1,\ldots,R$ and
  $\psi_j(X,T;\theta,m) = (1-T) \exp\{\sum_{k=1}^p \theta_k c_k(X)\} (c_j(X) -
  m_j)$, then \eqref{eq:orig-eb-dual-derivative} is equivalent to
  \begin{equation}
    \label{eq:eq:orig-eb-dual-ee}
    \begin{aligned}
      \sum_{i=1}^n \phi_j(X_i,T_i;m) = 0,~j=1,\ldots,R, \\
      \sum_{i=1}^n \psi_j(X_i,T_i;\theta, m) = 0,~j=1,\ldots,R. \\
    \end{aligned}
  \end{equation}
  Since $\phi(\cdot)$ and $\psi(\cdot)$ are all smooth functions of
  $\theta$ and $m$, all we need to verify is that $m^{*}_j =
  \mathrm{E}[c_j(X)|T=1]$ and  $\theta^{*}$ is
  the unique solution to the population version of
  \eqref{eq:eq:orig-eb-dual-ee}. It is obvious that
  $m^*$ is the solution to $ \mathrm{E}[\phi_j(X,T;m)]
  = 0,~j=1,\ldots,R$. Now take conditional expectation of $\psi_j$ given $X$:
\[
\begin{aligned}
\mathrm{E}[\psi_j(X, T; \theta, m^{*})\,|\,X]
 &= (1 - e(X)) e^{\sum_{k=1}^p \theta_k c_k(X) } (c_j(X) -
  m_j^*)  \\
&= \left(1 - \frac{e^{\sum_{k=1}^p \theta_k^{*} c_k(X)}}{1 +
    e^{\sum_{k=1}^p \theta_k^{*} c_k(X)}}\right) e^{\sum_{k=1}^p \theta_k
  c_k(X)} (c_j(X) -
  m_j^{*}) \\
&= \frac{e^{\sum_{k=1}^p \theta_k c_k(X)}}{1 +
    e^{\sum_{k=1}^p \theta_k^{*} c_k(X)}}  (c_j(X) - \mathrm{E}[c_j(X)|T=1]).
\end{aligned}
\]
The only way to make $\mathrm{E}[\psi_j(X, T; \hat{\theta}, m^{*})] = 0$
is to have $$\frac{e^{\sum_{k=1}^p \hat{\theta}_k c_k(X)}}{1 +
    e^{\sum_{k=1}^p \theta_k^{*} c_k(X)}} = \mathrm{const} \cdot \mathrm{P}(T = 1|X),$$
i.e.\ $\hat{\theta} = \theta^{*}$. This proves the consistency of $\hat{\theta}^{\textrm{EB}}$.

The consistency of
  $\hat{\gamma}^{\textrm{EB}}$ is proved by noticing
\[
w_i^{\textrm{EB}} = \frac{\exp(\sum_{j=1}^p \hat{\theta}_j^{\textrm{EB}} c_j(X_i))}{\sum_{T_i=0}\exp(\sum_{j=1}^p
    \hat{\theta}_j^{\textrm{EB}} c_j(X_i))} \overset{p}{\to}
  \frac{P(T_i=1|X_i)}{1 - P(T_i=1|X_i)},
\]
which is the IPW-NR weight defined in \eqref{eq:gamma-ipw}.
\end{proof}

The second claim is a corollary of Theorem \ref{thm:variance}, which
is proved below. For simplicity
we denote $\xi = (m^T, \theta^T, \mu(1|1), \gamma)^T$ and the true parameter
as $\xi^{*}$. Throughout this section we assume
$\mathrm{logit}(e(X)) = \sum_{j=1}^p \theta_j^{*}
c_j(X)$. Denote $\tilde{c}(X) = c(X) - \bar{c}^*(1)$, $e^{*}(X) = e(X;
\theta^{*})$, $l^{*}(X) =
\exp\{\sum_{j=1}^p \theta^{*}_j c_j(X)\} = e^{*}(X) / (1 -
e^{*}(X))$. Let
\[
  \begin{split}
\phi_j(X,T;m) &= T(c_j(X) - m_j),~j=1,\ldots,p,\\
\psi_j(X,T;\theta,m) &= (1-T) e^{\sum_{k=1}^p \theta_k c_k(X)} (c_j(X) -
  m_j), j = 1,\dotsc,p,\\
\varphi_{1|1}(X,T,Y;\mu(1|1)) &= T (Y - \mu(1|1)),\\
\varphi(X,T,Y;\theta, \mu(1|1), \gamma) &= (1-T)e^{\sum_{j=1}^p
    \theta_j c_j(X)} (Y + \gamma - \mu(1|1)), \\
  \end{split}
\] and $\zeta(X,T,Y;m,\theta,\mu(1|1),\gamma) = (\phi^T,
\psi^T, \varphi_{1|1}, \varphi)^T$ be all the estimating
equations. The Entropy Balancing estimator $\hat{\gamma}^{\textrm{EB}}$ is the solution to
\begin{equation}
  \label{eq:eb-entire-ee}
  \frac{1}{n} \sum_{i=1}^{n} \zeta(X_i,T_i,Y_i;m,\theta,\mu(1|1),\gamma) = 0.
\end{equation}

There are two forms of ``information" matrix that need to be
computed. The first is
\[
\footnotesize
\arraycolsep=3pt
\medmuskip = 1mu
\begin{aligned}
&A^{\textrm{EB}}(\xi^{*}) = \mathrm{E}\left[ - \frac{\partial}{\partial \xi^T}
  \zeta(X,T,Y;\xi^{*}) \right] \\
&=
\left(\mathrm{E}\left[ - \frac{\partial}{\partial m^T}
  \zeta(\xi^{*}) \right] ~
\mathrm{E}\left[ - \frac{\partial}{\partial \theta^T}
  \zeta(\xi^{*}) \right] ~
\mathrm{E}\left[ - \frac{\partial}{\partial \mu(1|1)}
  \zeta(\xi^{*}) \right] ~
\mathrm{E}\left[ - \frac{\partial}{\partial \gamma}
  \zeta(\xi^{*}) \right] \right) \\
&=
\mathrm{E} \left[
\begin{pmatrix}
  T \cdot I_R & 0 & 0 & 0 \\
  (1-T) l^{*}(X) \cdot I_R & -
  (1-T) l^{*}(X) (c(X) - \bar{c}^*(1)) c(X)^{T} & 0 &
  0 \\
  0^T & 0^T & T & 0 \\
  0 & - (1-T) l^{*}(X) (Y(0) - \mu^{*}(0|1))
  c(X)^T & (1-T) l^{*}(X) & - (1-T) l^{*}(X) \\
\end{pmatrix} \right] \\
&=
\pi \cdot
\begin{pmatrix}
  I_R & 0 & 0 & 0 \\
  I_R & - \mathrm{Cov}[c(X)|T=1] & 0 & 0 \\
  0^T & 0^T & 1 & 0 \\
  0 &  - \mathrm{Cov}(Y(0),c(X)|T=1) & 1 & -1 \\
\end{pmatrix}.
\end{aligned}
\]

A very useful identity in the computation of the expectation is
\[
\begin{aligned}
\mathrm{E}[f(X,Y)|T=1]
&= \pi^{-1} \mathrm{E}[e(X) f(X,Y)] \\
&= \frac{\mathrm{P}(T=0)}{\pi} \cdot \mathrm{E}\left[
  \frac{e(x)}{1-e(X)} f(X,Y) \middle| T=0 \right].
\end{aligned}
\]

The second information matrix is the covariance of
$\zeta(X,T,Y;\xi^{*})$. Denote $\tilde{Y}(t) = Y(t) - \mu^{*}(t|1),~t=0,1$
\[
\footnotesize
\arraycolsep=3pt
\medmuskip = 1mu
\begin{aligned}
  &B^{\textrm{EB}}(\xi^{*}) = \mathrm{E}[\zeta(X,Y,T;\xi^{*})
  \zeta(X,Y,T;\xi^{*})^T] \\
  &= \mathrm{E} \left[
    \begin{pmatrix}
      T \tilde{c}(X) \tilde{c}(X)^T & 0  & T \tilde{Y}(1) \tilde{c}(X)  & 0 \\
      0 & (1-T) l^{*}(X)^2 \tilde{c}(X)
      \tilde{c}(X)^{T} & 0 & (1-T) l^{*}(X)^2
      \tilde{Y}(0) \tilde{c}(X)\\
      T \tilde{Y}(1) \tilde{c}(X)^T & 0^T & T \tilde{Y}^2(1) & 0
\\
      0 & (1-T) l^{*}(X)^2
      \tilde{Y}(0) \tilde{c}(X)^T & 0 & (1-T) l^{*}(X)^2 \tilde{Y}^2(0) \\
    \end{pmatrix} \right]. \\
\end{aligned}
\]

The asymptotic distribution of $\hat{\gamma}^{\textrm{EB}}$ is
$\mathrm{N}(\gamma, V^{\textrm{EB}}(\xi^{*})/n)$ where $V^{\textrm{EB}}(\xi^{*})$ is the
bottom right entry of $A^{\textrm{EB}}(\xi^{*})^{-1} B^{\textrm{EB}}(\xi^{*})
A^{\textrm{EB}}(\xi^{*})^{-T}$. Let's denote
\[H_{a_1,a_2} = \mathrm{Cov}(a_1,a_2
| T=1),\]
\[G_{a_1,a_2} = \mathrm{E}\left[l^{*}(X)(a_1 - \mathrm{E}[a_1|T=1])
  (a_2 - \mathrm{E}[a_2|T=1])^T|T=1\right]^T,\]
and $H_a = H_{a,a}$, $G_a = G_{a,a}$. So
\[
A^{\textrm{EB}}(\xi^{*}) = \pi \cdot
\begin{pmatrix}
  I_R & 0 & 0 & 0 \\
  I_R & - H_{c(X)} & 0 & 0 \\
  0^T & 0^T & 1 & 0 \\
  0 &  - H_{Y(0),c(X)} & 1 & -1 \\
\end{pmatrix},
\]
\[
A^{\textrm{EB}}(\xi^{*})^{-1} = \pi^{-1} \cdot
\begin{pmatrix}
  I_R & 0 & 0 & 0 \\
  H_{c(X)}^{-1} & - H_{c(X)}^{-1} & 0 & 0 \\
  0^T & 0^T & 1 & 0 \\
  - H_{c(X),Y(0)}^T H_{c(X)}^{-1} &  H_{c(X), Y(0)}^T H_{c(X)}^{-1} & 1 &
  -1 \\
\end{pmatrix},
\]
and
\[
B^{\textrm{EB}}(\xi^{*}) =
\pi \cdot
  \begin{pmatrix}
    H_{c(X)} & 0 &H_{c(X), Y(1)} & 0 \\
    0 & G_{c(X)} & 0 & G_{ c(X), Y(0)} \\
    H_{Y(1), c(X)} & 0^T & H_{Y(1)} & 0 \\
    0 & G_{Y(0), c(X)}^T & 0 & G_{Y(0)} \\
  \end{pmatrix}.
\]

Thus
\[
\begin{aligned}
V^{\textrm{EB}}
= \pi^{-1} 
\cdot \left\{H_{c,0}^T H_{c}^{-1} \left(
H_{c,0} + G_{c}
  H_{c}^{-1} H_{c,0} - 2 G_{c,0} - 2 H_{c,1} \right)  + H_{1} +  G_{0}\right\}
\end{aligned}.
\]

It would be interesting to compare $V^{\textrm{EB}}(\xi^{*})$ with
$V^{\textrm{IPW}}(\xi^{*})$, the asymptotic variance of
$\hat{\gamma}^{\textrm{IPW}}$. The IPW PATT estimator
\eqref{eq:gamma-ipw} is equivalent to solving the following estimating
equations
\[
\sum_{i=1}^n \left(T_i - \frac{1}{1+e^{-\sum_{k=1}^p \theta_k c_k(X_i)}}\right) c_j(X_i) = 0,~r = 1,\ldots,R,
\]
\[
  \frac{1}{n} \sum_{i=1}^{n} \varphi_{1|1}(X_i,T_i,Y_i;\theta,\mu(1|1),\gamma) = 0,
\]
\[
  \frac{1}{n} \sum_{i=1}^{n} \varphi(X_i,T_i,Y_i;\theta,\mu(1|1),\gamma) = 0.
\]

If we call $K_{a_1, a_2} = \mathrm{E}[(1-e(X)) a_1 a_2^T | T = 1]$, we have
\[
\footnotesize
\arraycolsep=3pt
\medmuskip = 1mu
\begin{aligned}
A^{\textrm{IPW}}(\xi^{*}) &=
\mathrm{E}\left[
  \begin{pmatrix}
    e^{*}(X) (1-e^{*}(X)) c(X) c(X)^T & 0 & 0 \\
    0^T & T & 0 \\
    - (1-T) l^{*}(X) \tilde{Y}(0)
  c(X)^T & (1-T) l^{*}(X) & - (1-T) l^{*}(X) \\
  \end{pmatrix}
\right] \\
&= \pi \cdot
  \begin{pmatrix}
    K_{c(X)} & 0 & 0 \\
    0^T & 1 & 0 \\
    - H_{Y(0),c(X)} & 1& - 1 \\
  \end{pmatrix}. \\
A^{\textrm{IPW}}(\xi^{*})^{-1} &= \pi^{-1} \cdot
  \begin{pmatrix}
    K_{c(X)}^{-1} & 0 & 0 \\
    0^T & 1 & 0 \\
    - H_{Y(0),c(X)} K_{c(X)}^{-1} & 1& - 1 \\
  \end{pmatrix}.
\end{aligned}
\]
Let $q^{*}(X) = e^{*}(X) l^{*}(X)$,
\[
\footnotesize
\arraycolsep=2pt
\medmuskip = 1mu
\begin{aligned}
B^{\textrm{IPW}}(\xi^{*}) &=
\mathrm{E}\left[
  \begin{pmatrix}
    (T-e^{*}(X))^2 c(X) c(X)^T & T (T - e^{*}(X)) \tilde{Y}(1) c(X) &
    - (1-T) q^{*}(X) \tilde{Y}(0) c(X) \\
    T (T - e^{*}(X)) \tilde{Y}(1) c(X)^T & T \tilde{Y}^2(1) & 0 \\
    - (1-T)q^{*}(X) \tilde{Y}(0) c(X) & 0 & (1-T) l^{*}(X)^2 \tilde{Y}^2(0) \\
  \end{pmatrix}
\right] \\
&= \pi \cdot
\begin{pmatrix}
  K_{c(X)} & K_{c(X), \tilde{Y}(1)} & K_{c(X), \tilde{Y}(0)} - H_{c(X), Y(0)}
  \\
  K_{c(X), \tilde{Y}(1)}^T & H_{Y(1)} & 0 \\
  K_{c(X), \tilde{Y}(0)}^T - H_{c(X), Y(0)}^T & 0 & G_{Y(0)} \\
\end{pmatrix}.
\end{aligned}
\]

$V^{\textrm{IPW}}$ can thus be computed consequently and the details
are omitted.

\bibliographystyle{chicago}
\bibliography{../reference/ref}

\begin{thebibliography}{}

\bibitem[\protect\citeauthoryear{Abadie and Imbens}{Abadie and
  Imbens}{2006}]{abadie2006large}
Abadie, A. and G.~W. Imbens (2006).
\newblock Large sample properties of matching estimators for average treatment
  effects.
\newblock {\em Econometrica\/}~{\em 74\/}(1), 235--267.

\bibitem[\protect\citeauthoryear{Albert and Andersen}{Albert and
  Andersen}{1984}]{Albert1984}
Albert, A. and J.~A. Andersen (1984).
\newblock On the existence of maximum likelihood estimates in logistic
  regression models.
\newblock {\em Biometrika\/}~{\em 71\/}(1), 1--10.

\bibitem[\protect\citeauthoryear{Bang and Robins}{Bang and
  Robins}{2005}]{Bang2005}
Bang, H. and J.~M. Robins (2005).
\newblock Doubly robust estimation in missing data and causal inference models.
\newblock {\em Biometrics\/}~{\em 61\/}(4), 962--973.

\bibitem[\protect\citeauthoryear{Chan, Yam, and Zhang}{Chan
  et~al.}{2015}]{chan2015}
Chan, K. C.~G., S.~C.~P. Yam, and Z.~Zhang (2015).
\newblock Globally efficient nonparametric inference of average treatment
  effects by empirical balancing calibration weighting.
\newblock {\em Journal of Royal Statistical Society, Series B
  (Methodology)\/}~{\em to appear}.

\bibitem[\protect\citeauthoryear{Cover and Thomas}{Cover and
  Thomas}{2012}]{cover2012}
Cover, T.~M. and J.~A. Thomas (2012).
\newblock {\em Elements of information theory}.
\newblock John Wiley \& Sons.

\bibitem[\protect\citeauthoryear{Deville and S{\"a}rndal}{Deville and
  S{\"a}rndal}{1992}]{deville1992calibration}
Deville, J.-C. and C.-E. S{\"a}rndal (1992).
\newblock Calibration estimators in survey sampling.
\newblock {\em Journal of the American Statistical Association\/}~{\em
  87\/}(418), 376--382.

\bibitem[\protect\citeauthoryear{Diamond and Sekhon}{Diamond and
  Sekhon}{2013}]{Diamond2013}
Diamond, A. and J.~S. Sekhon (2013).
\newblock Genetic matching for estimating causal effects: A general
  multivariate matching method for achieving balance in observational studies.
\newblock {\em Review of Economics and Statistics\/}~{\em 95\/}(3), 932--945.

\bibitem[\protect\citeauthoryear{Ferwerda}{Ferwerda}{2014}]{Ferwerda2014}
Ferwerda, J. (2014).
\newblock Electoral consequences of declining participation: A natural
  experiment in austria.
\newblock {\em Electoral Studies\/}~{\em 35\/}(0), 242--252.

\bibitem[\protect\citeauthoryear{Gneiting and Raftery}{Gneiting and
  Raftery}{2007}]{gneiting2007strictly}
Gneiting, T. and A.~E. Raftery (2007).
\newblock Strictly proper scoring rules, prediction, and estimation.
\newblock {\em Journal of the American Statistical Association\/}~{\em
  102\/}(477), 359--378.

\bibitem[\protect\citeauthoryear{Graham, Pinto, and Egel}{Graham
  et~al.}{2012}]{graham2012inverse}
Graham, B.~S., C.~C. D.~X. Pinto, and D.~Egel (2012).
\newblock Inverse probability tilting for moment condition models with missing
  data.
\newblock {\em The Review of Economic Studies\/}~{\em 79\/}(3), 1053--1079.

\bibitem[\protect\citeauthoryear{Hahn}{Hahn}{1998}]{Hahn1998}
Hahn, J. (1998).
\newblock On the role of the propensity score in efficient semiparametric
  estimation of average treatment effects.
\newblock {\em Econometrica\/}~{\em 66\/}(2), 315--332.

\bibitem[\protect\citeauthoryear{Hainmueller}{Hainmueller}{2011}]{Hainmueller2011}
Hainmueller, J. (2011).
\newblock Entropy balancing for causal effects: A multivariate reweighting
  method to produce balanced samples in observational studies.
\newblock {\em Political Analysis\/}~{\em 20}, 25--46.

\bibitem[\protect\citeauthoryear{Hastie and Tibshirani}{Hastie and
  Tibshirani}{1990}]{hastie1990generalized}
Hastie, T.~J. and R.~J. Tibshirani (1990).
\newblock {\em Generalized additive models}, Volume~43.
\newblock CRC Press.

\bibitem[\protect\citeauthoryear{Hirano and Imbens}{Hirano and
  Imbens}{2001}]{Hirano2001}
Hirano, K. and G.~Imbens (2001).
\newblock Estimation of causal effects using propensity score weighting: An
  application to data on right heart catheterization.
\newblock {\em Health Services and Outcomes Research Methodology\/}~{\em 2},
  259--278.

\bibitem[\protect\citeauthoryear{Hirano, Imbens, and Ridder}{Hirano
  et~al.}{2003}]{Hirano2003}
Hirano, K., G.~W. Imbens, and G.~Ridder (2003).
\newblock Efficient estimation of average treatment effects using the estimated
  propensity score.
\newblock {\em Econometrica\/}~{\em 71\/}(4), 1161--1189.

\bibitem[\protect\citeauthoryear{Holland}{Holland}{1986}]{Holland1986}
Holland, P.~W. (1986).
\newblock Statistics and causal inference.
\newblock {\em Journal of the American Statistical Association\/}~{\em 81},
  945--960.

\bibitem[\protect\citeauthoryear{Imai, King, and Stuart}{Imai
  et~al.}{2008}]{Imai2008}
Imai, K., G.~King, and E.~A. Stuart (2008).
\newblock Misunderstandings between experimentalists and observationalists
  about causal inference.
\newblock {\em Journal of the Royal Statistical Society: Series A (Statistics
  in Society)\/}~{\em 171\/}(2), 481--502.

\bibitem[\protect\citeauthoryear{Imai and Ratkovic}{Imai and
  Ratkovic}{2014}]{Imai2014}
Imai, K. and M.~Ratkovic (2014).
\newblock Covariate balancing propensity score.
\newblock {\em Journal of the Royal Statistical Society: Series B (Statistical
  Methodology)\/}~{\em 76\/}(1), 243--263.

\bibitem[\protect\citeauthoryear{Kang and Schafer}{Kang and
  Schafer}{2007}]{Kang2007}
Kang, J.~D. and J.~L. Schafer (2007).
\newblock Demystifying double robustness: A comparison of alternative
  strategies for estimating a population mean from incomplete data.
\newblock {\em Statistical Science\/}~{\em 22\/}(4), 523--539.

\bibitem[\protect\citeauthoryear{Lunceford and Davidian}{Lunceford and
  Davidian}{2004}]{Lunceford2004}
Lunceford, J.~K. and M.~Davidian (2004).
\newblock Stratification and weighting via the propensity score in estimation
  of causal treatment effects: a comparative study.
\newblock {\em Statistics in Medicine\/}~{\em 23\/}(19), 2937--2960.

\bibitem[\protect\citeauthoryear{Marcus}{Marcus}{2013}]{Marcus2013}
Marcus, J. (2013).
\newblock The effect of unemployment on the mental health of spouses –
  evidence from plant closures in germany.
\newblock {\em Journal of Health Economics\/}~{\em 32\/}(3), 546--558.

\bibitem[\protect\citeauthoryear{Neyman}{Neyman}{1923}]{neyman1923applications}
Neyman, J. (1923).
\newblock Sur les applications de la thar des probabilities aux experiences
  agaricales: Essay des principle. excerpts reprinted (1990) in english.
\newblock {\em Statistical Science\/}~{\em 5}, 463--472.

\bibitem[\protect\citeauthoryear{Qin and Zhang}{Qin and Zhang}{2007}]{Qin2007}
Qin, J. and B.~Zhang (2007).
\newblock Empirical-likelihood-based inference in missing response problems and
  its application in observational studies.
\newblock {\em Journal of the Royal Statistical Society: Series B (Statistical
  Methodology)\/}~{\em 69\/}(1), 101--122.

\bibitem[\protect\citeauthoryear{Ridgeway and McCaffrey}{Ridgeway and
  McCaffrey}{2007}]{Ridgeway2007}
Ridgeway, G. and D.~F. McCaffrey (2007).
\newblock Comment: Demystifying double robustness: A comparison of alternative
  strategies for estimating a population mean from incomplete data.
\newblock {\em Statistical Science\/}~{\em 22\/}(4), 540--543.

\bibitem[\protect\citeauthoryear{Robins, Sued, Lei-Gomez, and Rotnitzky}{Robins
  et~al.}{2007}]{Robins2007}
Robins, J., M.~Sued, Q.~Lei-Gomez, and A.~Rotnitzky (2007).
\newblock Comment: Performance of double-robust estimators when “inverse
  probability” weights are highly variable.
\newblock {\em Statistical Science\/}~{\em 22\/}(4), 544--559.

\bibitem[\protect\citeauthoryear{Robins, Rotnitzky, and Zhao}{Robins
  et~al.}{1994}]{Robins1994}
Robins, J.~M., A.~Rotnitzky, and L.~Zhao (1994).
\newblock Estimation of regression coefficients when some regressors are not
  always observed.
\newblock {\em Journal of the American Statistical Association\/}~{\em 89},
  846--866.

\bibitem[\protect\citeauthoryear{Rosenbaum and Rubin}{Rosenbaum and
  Rubin}{1983}]{rosenbaum1983}
Rosenbaum, P. and D.~Rubin (1983).
\newblock The central role of the propensity score in observational studies for
  causal effects.
\newblock {\em Biometrika\/}~{\em 70\/}(1), 41--55.

\bibitem[\protect\citeauthoryear{Rosenbaum and Rubin}{Rosenbaum and
  Rubin}{1984}]{Rosenbaum1984}
Rosenbaum, P. and D.~Rubin (1984).
\newblock Reducing bias in observational studies using subclassification on the
  propensity score.
\newblock {\em Journal of the American Statistical Association\/}~{\em 79},
  516--524.

\bibitem[\protect\citeauthoryear{Rosenbaum and Rubin}{Rosenbaum and
  Rubin}{1985}]{rosenbaum1985constructing}
Rosenbaum, P.~R. and D.~B. Rubin (1985).
\newblock Constructing a control group using multivariate matched sampling
  methods that incorporate the propensity score.
\newblock {\em The American Statistician\/}~{\em 39\/}(1), 33--38.

\bibitem[\protect\citeauthoryear{Rubin}{Rubin}{1974}]{Rubin1974}
Rubin, D. (1974).
\newblock Estimating causal effects of treatments in randomized and
  nonrandomized studies.
\newblock {\em Journal of Educational Psychology\/}~{\em 66\/}(5), 688--701.

\bibitem[\protect\citeauthoryear{S{\"a}rndal and Lundstr{\"o}m}{S{\"a}rndal and
  Lundstr{\"o}m}{2005}]{sarndal2005estimation}
S{\"a}rndal, C.-E. and S.~Lundstr{\"o}m (2005).
\newblock {\em Estimation in surveys with nonresponse}.
\newblock John Wiley \& Sons.

\bibitem[\protect\citeauthoryear{Silvapulle}{Silvapulle}{1981}]{Silvapulle1981}
Silvapulle, M.~J. (1981).
\newblock On the existence of maximum likelihood estimators for the binomial
  response models.
\newblock {\em Journal of the Royal Statistical Society. Series B
  (Methodological)\/}~{\em 43\/}(3), 310--313.

\bibitem[\protect\citeauthoryear{Stefanski and Boos}{Stefanski and
  Boos}{2002}]{Stefanski2002}
Stefanski, L.~A. and D.~D. Boos (2002).
\newblock {The Calculus of M-Estimation}.
\newblock {\em The American Statistician\/}~{\em 56\/}(1), 29--38.

\bibitem[\protect\citeauthoryear{Tan}{Tan}{2006}]{Tan2006}
Tan, Z. (2006).
\newblock A distributional approach for causal inference using propensity
  scores.
\newblock {\em Journal of the American Statistical Association\/}~{\em 101},
  1619--1637.

\bibitem[\protect\citeauthoryear{Tan}{Tan}{2010}]{Tan2010}
Tan, Z. (2010).
\newblock Bounded, efficient and doubly robust estimation with inverse
  weighting.
\newblock {\em Biometrika\/}~{\em 97\/}(3), 661--682.

\bibitem[\protect\citeauthoryear{Wang and Rao}{Wang and Rao}{2002}]{Wang2002}
Wang, Q. and J.~N.~K. Rao (2002).
\newblock Empirical likelihood-based inference under imputation for missing
  response data.
\newblock {\em The Annals of Statistics\/}~{\em 30\/}(3), 896--924.

\bibitem[\protect\citeauthoryear{Zhao}{Zhao}{2016}]{zhao2016covariate}
Zhao, Q. (2016).
\newblock Covariate balancing propensity score by tailored loss functions.

\bibitem[\protect\citeauthoryear{Zubizarreta}{Zubizarreta}{2015}]{zubizarreta2015stable}
Zubizarreta, J.~R. (2015).
\newblock Stable weights that balance covariates for estimation with incomplete
  outcome data.
\newblock {\em Journal of the American Statistical Association\/}~{\em
  110\/}(511), 910--922.

\end{thebibliography}

\end{document}